\documentclass[10pt, conference]{IEEEtran}
\usepackage{amsmath,amsthm,amsfonts}
\usepackage{graphicx}
\usepackage{epsf}
\usepackage{mdwmath}
\usepackage{mdwtab}
\usepackage{cite}
\usepackage{amssymb}
\usepackage{algorithmic}
\usepackage{algorithm}
\usepackage{amssymb}
\usepackage{algorithmic}
\usepackage{algorithm}
\usepackage{bm}
\usepackage{color}
\newtheorem{thm}{Theorem}

\theoremstyle{remark}
\newtheorem{remark}[thm]{Remark}

\def\R{\mathbb{R}}

\def\1{\bm{1}}

\title{Near-Optimal Compressive Binary Search}

\begin{document}
\author{\IEEEauthorblockN{Matthew L. Malloy}
\IEEEauthorblockA{ Electrical and Computer Engineering\\
University of Wisconsin-Madison\\
Email: mmalloy@wisc.edu}
\and
\IEEEauthorblockN{Robert D. Nowak}
\IEEEauthorblockA{
Electrical and Computer Engineering\\
University of Wisconsin-Madison\\
Email: nowak@ece.wisc.edu}
}
\maketitle

\begin{abstract}
We propose a simple modification to the recently proposed \emph{compressive binary search} \cite{2012arXiv1202.0937D}.  
The modification removes an unnecessary and suboptimal factor of $\log\log n$ from the SNR requirement, making the procedure optimal (up to a small constant).  Simulations show that the new procedure performs significantly better in practice as well.
We also contrast this problem with the more well known problem of {\em noisy binary search}. 
\end{abstract}
\section{Introduction}
The recently proposed \emph{compressive binary search} (CBS) algorithm \cite{2012arXiv1202.0937D} aims to determine the location of a single non-zero entry in a vector $\bm{x} \in \mathbb{R}^n$ using $m$ adaptive linear projections of the form 
\begin{eqnarray} \nonumber
y_i = \left\langle \bm{a}_i, \bm{x}  \right\rangle + z_i \qquad i = 1,...,m
\end{eqnarray}
where $\bm{a}_i\in\R^n$ are sensing vectors with $||\bm{a}_i|| =1$ and $z_i$ are i.i.d. $\mathcal{N}(0,1)$. The CBS algorithm proposed in \cite{2012arXiv1202.0937D} finds the non-zero entry with vanishing probability of error as $n$ gets large provided
$\mu \geq C \sqrt{(n/m) \log \log_2 n}$
where $\mu$ is the amplitude of the non-zero entry.
A bit more precisely, Theorem 1 in \cite{2012arXiv1202.0937D} states that $\mathbb{P}_e \leq \delta$ provided
\begin{eqnarray} \label{eqn:CBSOrig}
\mu \geq  \sqrt{ \frac{8n}{m}\left( \log \frac{1}{2\delta} + \log \log_2 n \right)}
\end{eqnarray} 
where $\mathbb{P}_e$ is the probability the procedure fails to return an index corresponding to the non-zero entry.  The dependence on $\delta$ is to be expected, but the authors of \cite{2012arXiv1202.0937D} rightly question whether the $\log\log_2 n$ term is needed.  The main contribution of this paper is a simple modification of the algorithm proposed in \cite{2012arXiv1202.0937D} which eliminates this unnecessary and suboptimal dependence on $n$.

\section{Main Result} 
The CBS algorithm proposed in \cite{2012arXiv1202.0937D} operates as follows.
The algorithm consists of $s_0 = \log_2 n$ steps (as in \cite{2012arXiv1202.0937D} we assume $n$ is dyadic). At each step of the algorithm, indexed by $s$, measurements are taken on progressively smaller dyadic subintervals of $\{1,\dots,n\}$.  The sensing vectors are discrete Haar wavelets, and the sign of each measurement gives an indication of whether the non-zero entry is in the right or left half-interval of the wavelet's support.   The CBS algorithm is outlined in Fig.~\ref{fig:ocbs}. 

\begin{figure}[h] \small
\fbox{\parbox[b]{3.35in}{{\underline{\bf Compressive Binary Search (CBS)}}  \\ \vspace{-.15in} \\
number of steps: $s_0 := \log_2 n$ \\
measurements per step: 
$m_s$ , with $\sum_{s=1}^{s_0} m_s \leq m$ \\
initial support: $J_0^{(1)} := \{1,\dots,n\}$ \\ \vspace{-.1in} \\
for $s=1,\dots,s_0$ \\ \vspace{-.15in} \\
1) split: $J_1^{(s)}$ and $J_2^{(s)}$, left and right subinterval of $J_0^{(s)}$ \\
2) sensing vector: $\bm{u}^{(s)} = 2^{-(s_0-s+1)/2}$ on $J_1^{(s)}$, \\ \indent \ \ \ \ $\bm{u}^{(s)}= -2^{-(s_0-s+1)/2}$ on $J_2^{(s)}$, and $0$ otherwise  \\
3) measure: $y_i^{(s)} = \langle \bm{u}^{(s)},\bm{x}\rangle + z_i^{(s)}, \\
\indent \ \ \ \ z_i^{(s)} \overset{iid}{\sim} {\cal N}(0,1), \  i=1,\dots,m_s$ \\
4) update support: $J_0^{(s+1)} = J_1^{(s)}$ if $\sum_{i=1}^{m_s}y_i^{(s)}>0$\\ \indent \ \ \ \  and $J_0^{(s+1)} = J_2^{(s)}$ otherwise \\
end  \\ \vspace{-.1in} \\
output: $J_0^{(s_0+1)}$ (a single index)
}}
\caption{Compressive Binary Search algorithm.}
\label{fig:ocbs}
\end{figure}

Since the sensing vectors have unit norm, the magnitude of the non-zero entries of the sensing vectors grow as the algorithm proceeds from coarse-to-fine wavelets. Consequently, the SNR of the measurements grows exponentially as the procedure progresses.  In \cite{2012arXiv1202.0937D}, the authors divide the $m$ measurements among the steps as follows.  For $s=1,\dots,s_0$
\begin{eqnarray*}
m_s := \tilde{m}_s +1, \qquad   \tilde{m}_s = \left \lfloor (m - s_0) \; 2^{-s} \right \rfloor
\end{eqnarray*}
It is easy to check that $\sum_{s=1}^{s_0} m_s \leq m$. This allocation roughly equalizes the SNR at each step. If $\mu$ satisfies (\ref{eqn:CBSOrig}), then the CBS algorithm is guaranteed to have a probability of error of at most $\delta/s_0$ at each step (and the union bound yields an overall probability of error of at most $\delta$).

Here, instead, we allocate the $m$ measurements as follows:
\begin{eqnarray*}
m_s := \tilde{m}_s +1, \qquad   \tilde{m}_s = \left \lfloor (m - s_0) \; s\; 2^{-(s+1)} \right \rfloor.
\end{eqnarray*}
Again, it is easily verified that $\sum_{s=1}^{s_0} m_s \leq m$. With this allocation, the probability of error is at most a constant times $2^{-s}$ at each step; i.e.,
the probability of error decreases exponentially over the steps, rather than remaining constant as above. This simple modification is enough to eliminate the $\log\log_2 n$ term in the bound in  \cite{2012arXiv1202.0937D}.
\begin{thm}
If $m \geq 2 \log_2 n$ and
\begin{eqnarray}
m_s := \tilde{m}_s +1, \qquad   \tilde{m}_s = \left \lfloor (m - s_0) \; s\; 2^{-(s+1)} \right \rfloor
\end{eqnarray} 
then $\sum_{s=1}^{s_0} m_s\leq m$ and the CBS algorithm succeeds with
$\mathbb{P}_e \leq \delta$
provided magnitude of the non-zero entry satisfies
\begin{eqnarray} \label{eqn:CBSimprov}
\mu \geq  \sqrt{ \frac{16n}{m}\log\left( \frac{1}{2\delta} +1 \right)}.
\end{eqnarray}
\end{thm}
\begin{proof}
The total measurement budget satisfies 
\begin{eqnarray*}
\sum_{s=1}^{s_0} m_s = s_0 + \sum_{s=1}^{s_0} \tilde{m}_s \leq s_0 + (m-s_0) \sum_{s=1}^{s_0} \; s \; 2^{-(s+1)} \leq m,
\end{eqnarray*}
since $\sum_{s=1}^{s_0}  s  2^{-(s+1)} \leq 1$.
By the union bound and a Gaussian tail bound, the total error probability satisfies 
\begin{eqnarray*}
\mathbb{P}_e &\leq &\frac{1}{2}\,  \sum_{s=1}^{s_0}  \exp \left( -\frac{m_s \mu^2 2^s}{4n}\right).
\end{eqnarray*}
Since $m_s \geq (m-s_0) s  2^{-(s+1)}$ and $m \geq 2 \log_2 n = 2 s_0$, we conclude $m_s \geq\frac{m}{2} s 2^{-(s+1)}$ and thus 
\begin{eqnarray*}
\mathbb{P}_e &\leq &\frac{1}{2} \, \sum_{s=1}^{s_0}  \exp \left( -\frac{m s\mu^2}{16n}\right).
\end{eqnarray*}
Letting $\mu \geq  \sqrt{\frac{16n}{m} \log\left( \frac{1}{2\delta} +1 \right)}$ yields
\begin{eqnarray*}
\mathbb{P}_e & \leq & \frac{1}{2} \, \sum_{s=1}^{s_0}   \left(\frac{1}{2\delta} +1\right)^{-s} \ \leq \ \delta \ .
\end{eqnarray*}
\end{proof}

\setcounter{thm}{0} 

\begin{remark}
The authors of \cite{2012arXiv1202.0937D}  also show that no procedure can succeed at the compressive binary search problem with probability greater than $1/2$ if $\mu \leq \sqrt{n/m}$.
Using the new allocation of measurements, the CBS algorithm succeeds with probability greater than $1/2$ if $\mu \geq \sqrt{16\log(2)n/m}$,  within a small constant factor of the lower bound.
\end{remark}

\begin{remark}
An algorithm similar to CBS was proposed in \cite{4786013} as a special case of a more general approach to the adaptive sensing problem, although control of the SNR across measurement steps is not discussed.  Additionally, the authors in \cite{5470144} suggest and analyze an adaptive group testing procedure with many of the main ideas of CBS.  The adaptive group testing procedure also has a sub-optimal $\log \log n$ dependence.
\end{remark}

\begin{remark}
The CBS problem is closely related to the so-called {\em noisy binary search} problem \cite{horstein,bz,gbs,kk}.  Noisy binary search addresses a version of the classic binary search problem with binary noise.  If the SNR is equal in each step of CBS (as in 
\cite{2012arXiv1202.0937D}), then it is equivalent to the ``naive'' noisy binary search algorithm discussed in \cite{kk}, which also suffers from the suboptimal $\log\log_2 n$ factor. 
More sophisticated algorithms such as Horstein's algorithm \cite{horstein,bz,gbs} and binary search with backtracking \cite{kk} are optimal to within constants.  The CBS problem, however, is different from noisy binary search in that more localized measurements (or ``queries'') are more reliable.  \newpage  \noindent Because of this unique feature, the optimal algorithms for noisy binary search do not yield optimal solutions for the CBS problem.  Instead, all that is needed to eliminate the $\log \log_2 n$ factor is a measurement allocation that properly exploits the fact that more localized measurements have a larger SNR.

\end{remark}


\begin{figure}[htb]
\vspace{.3cm}
\centerline{\includegraphics[width=9.6cm]{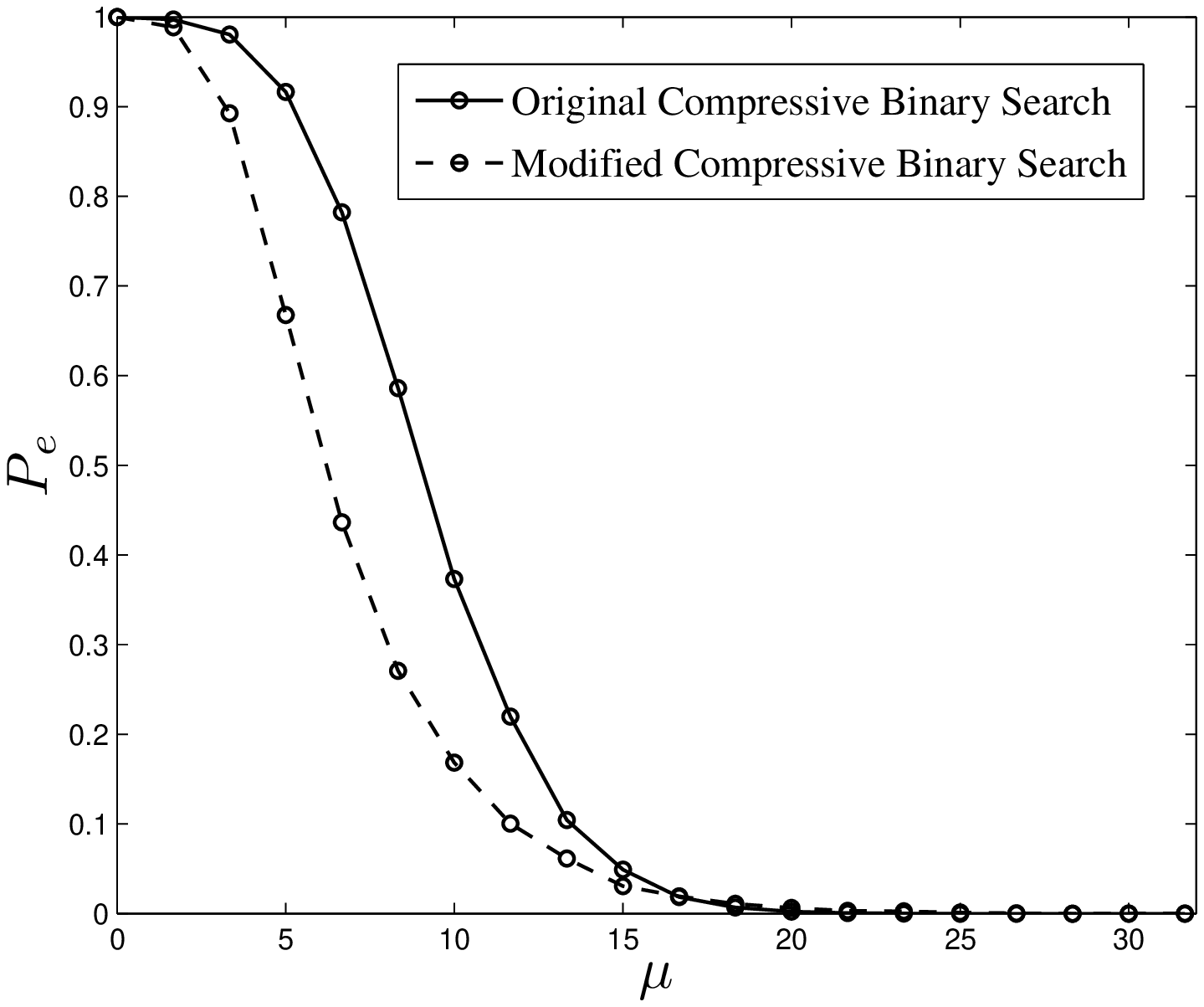}}
\caption{Numerical Simulation.  Empirical performance of CBS (both original and modified) as a function of $\mu$ for $n= 4096$ and $m = 256$. 10,000 trials. 
}
\vspace{.3cm}
\end{figure}

\bibliographystyle{IEEEtran}
\bibliography{CompBSnote.bib}

\end{document}